\documentclass[12pt,a4paper]{article}
\usepackage{mathrsfs}
\usepackage{amsmath}
\usepackage{dsfont}
\usepackage{amsfonts}
\usepackage{amssymb}
\usepackage{graphicx}
\usepackage{txfonts}

\newtheorem{fact}{Fact}
\newtheorem{definition}{Definition}
\newtheorem{example}{Example}
\newtheorem{proposition}{Proposition}

\newenvironment{proof}{\noindent\textit{Proof~~}}
{\nolinebreak[4]\hfill$\square$\\\par}

\title{A new model for quantum game based on the Marinatto-Weber approach}
\author{Piotr Fr\c{a}ckiewicz\\
\small Institute of Mathematics, Pomeranian University\\ \small 76-200 S\l upsk, Poland\\ \small P.Frackiewicz@impan.gov.pl}
\begin{document}
\maketitle
\begin{abstract}
The Marinatto-Weber approach to quantum game is a straightforward
way to apply the power of quantum mechanics to classical game
theory. In the simplest case, the quantum scheme is that players
manipulate their own qubits of a two-qubit state either with the
identity $\mathds{1}$ or the Pauli operator $\sigma_{x}$. However,
such a simplification of the scheme raises doubt as to whether it could
really reflect a quantum game. In this paper we put forward
examples which may constitute arguments against the present form
of the Marinatto-Weber scheme. Next, we modify the scheme to
eliminate the undesirable properties of the protocol by extending
the players' strategy sets.
\end{abstract}
\section{Introduction}
The Marinatto-Weber (MW) scheme \cite{mw} has become one of the
most frequently used schemes for quantum games. Though it was
created for research on Nash equilibria in quantum $2\times 2$
games, it has also found an application in studying some of the
refinements of a Nash equilibrium, such as evolutionarily stable
strategies \cite{iqbal, iqbal2}. Moreover, the MW scheme turns out
to be applicable to finite extensive games \cite{fr}. Among other
applications, the problem of duopoly is worthy noting. It was shown
in \cite{iqbalstack} that the MW scheme can be used in
Stackelberg's model of duopoly. The paper initiated further
studies on the quantum Stackelberg duopoly \cite{correlationoise,
krk1} and the quantum approach to Bertrand duopoly \cite{krk}. These
recent papers show uninterrupted interest in research on quantum
games played according to the MW idea and they provide sufficient
motivation to study this protocol.
\section{Preliminary to the MW approach}
The MW scheme was originally designed for a $2\times 2$ game
\begin{equation}\label{2x2game} \bordermatrix{&l & r \cr
                t& (a_{00},b_{00}) & (a_{01}, b_{01})\cr
                b& (a_{10},b_{10}) & (a_{11},b_{11})},
                ~~\mbox{where}~~a_{ij} \in \mathbb{R}.
\end{equation}
Each of the two players acts with the identity $\mathds{1}$ and
the Pauli operator $\sigma_{x}$ on his own qubit of some fixed
two-qubit state $\rho_{\mathrm{in}} = |\psi_{\mathrm{in}} \rangle
\langle \psi_{\mathrm{in}}|$, which is called the initial state.
The players' payoffs are determined by measurement of the
resulting final state $\rho_{\mathrm{fin}}$. Formally, the final
state takes the following form:
\begin{align}\label{staryfinalstate}
\rho_{\mathrm{fin}} &= pq\mathds{1}\otimes
\mathds{1}\rho_{\mathrm{in}}\mathds{1}\otimes \mathds{1} +
p(1-q)\mathds{1} \otimes \sigma_{x}\rho_{\mathrm{in}}\mathds{1}
\otimes \sigma_{x}\nonumber\\ &\quad + (1-p)q\sigma_{x} \otimes
\mathds{1}\rho_{\mathrm{in}} \sigma_{x} \otimes \mathds{1} +
(1-p)(1-q)\sigma_{x} \otimes \sigma_{x}\rho_{\mathrm{in}}
\sigma_{x} \otimes \sigma_{x},
\end{align}
where $p$ and $q$ are the probabilities of choosing the identity
$\mathds{1}$ by player 1 and player 2, respectively. Then the pair
of players' payoffs $(\pi_{1}, \pi_{2})$ depends on $p$ and $q$,
and through the measurement operator
\begin{equation}\label{payoffoperator}
X = (a_{00},b_{00})|00\rangle \langle 00| +
(a_{01},b_{01})|01\rangle \langle 01| + (a_{10},b_{10})|10\rangle
\langle 10| + (a_{11},b_{11})|11\rangle \langle 11|,
\end{equation}
is given by formula
\begin{equation}\label{payofffunctional}
(\pi_{1}, \pi_{2})(p,q) = \mathrm{tr}(X\rho_{\mathrm{fin}}).
\end{equation}
The MW protocol can
also be expressed by a useful diagram
\begin{equation}\label{diagram} \bordermatrix{&l & r \cr
                t& (a_{00},b_{00}) & (a_{01}, b_{01})\cr
                b& (a_{10},b_{10}) & (a_{11},b_{11})}
\xrightarrow{|\psi_{\mathrm{in}}\rangle =
\sum_{ij}\lambda_{ij}|ij\rangle} \bordermatrix{&\mathds{1} &
\sigma_{x} \cr
                \mathds{1}& (\alpha_{00},\beta_{00}) & (\alpha_{01}, \beta_{01})\cr
                \sigma_{x}& (\alpha_{10},\beta_{10}) & (\alpha_{11},\beta_{11})},
\end{equation}
where $|\psi_{\mathrm{in}}\rangle =
\sum_{ij}\lambda_{ij}|ij\rangle$ is a two-qubit state and
\begin{equation}\label{pary} \begin{array}{ll}
(\alpha_{00}, \beta_{00}) = \sum_{i,j}|\lambda_{i,j}|^2(a_{ij},
b_{ij}), \displaystyle & (\alpha_{01}, \beta_{01}) =
\sum_{i,j}|\lambda_{i,j\oplus_{2}1}|^2(a_{ij}, b_{ij});\\
(\alpha_{10}, \beta_{10}) =
\sum_{i,j}|\lambda_{i\oplus_{2}1,j}|^2(a_{ij}, b_{ij}),
\displaystyle  & (\alpha_{11}, \beta_{11}) =
\sum_{i,j}|\lambda_{i\oplus_{2}1,j\oplus_{2}1}|^2(a_{ij}, b_{ij}).
\end{array}
\end{equation}
The bimatrix on the left of diagram~(\ref{diagram}) represents the
classical (input) $2\times 2$ game while the bimatrix of the right
is the output game obtained by applying the MW protocol.
\paragraph{Advantages of the MW approach} The simplicity of the calculations
and the possibility to easily extend the scheme to consider more
complex games than $2\times 2$ can be counted among the
advantages of the MW scheme. The MW approach to $n\times m$
bimatrix games is defined by the initial state represented by a
vector from $\mathbb{C}^n\otimes\mathbb{C}^m$ and appropriate $n$
and $m$ unitary operators for player 1 and 2, respectively (see
\cite{commentfracor}). In any such case the dimension of the
output game is always equal to the input game, which does not make
the output game more difficult to deal with. In particular,
determining Nash equilibria or evolutionarily stable strategies
have the same level of difficulty in both the classical and
quantum games.

The MW approach has also achieved popularity because it has found
application in infinite games; that is in games where the players'
strategy sets are infinite. For example, in
\cite{iqbalstack} the authors showed a way to apply the MW
approach to Stackelberg's model of duopoly---game in which the
players' strategy sets are identified with interval $[0, \infty)$
(to learn the model of duopoly, see the orginal paper \cite{stack}
or textbook \cite{peters}).
\section{Undesirable features of
the MW approach}\label{disadvantages} In spite of the simplicity
and large number of applications,
 the MW scheme has been driven by more complex schemes: the
Eisert-Wilkens-Lewenstein approach \cite{ewl} in the case of
quantum $2\times 2$ games and the Li et al. approach
\cite{ciaglyduopol} in the case of various types of duopolies.
Below we point out the main flaws in the MW approach.
\paragraph{Problem of non-classical games induced by separable initial states} The common criticism of the MW scheme is
that the initial state has an excessive impact on the output game.
In other words, the initial state plays the main role in producing
the non-classical output game and the output games may vary
significantly from each other depending on the initial state.
However, in our view, this is not a disadvantage of the scheme.
Quantum information techniques like superdense coding, the Greenberger--Horne--Zeilinger
example, the violation of the Clauser--Horne--Shimony--Holt inequality and many others are
also based on preparing some special quantum state and (or)
operators acting on it so that non-classical results could be
obtained. But in every such protocol the quantum state has to be
an entangled one. Unfortunately, the MW scheme outputs a
non-classical game even if the initial state is separable. Let us
take the following example:
\begin{equation} \label{diagram1}
\bordermatrix{&l & r \cr
                t& (5,3) & (1,1)\cr
                b& (1,1) & (3,5)} \xrightarrow{|\psi_{\mathrm{in}}\rangle =
\frac{1}{\sqrt{2}}(|00\rangle + |01\rangle)}
\bordermatrix{&\mathds{1} & \sigma_{x} \cr \mathds{1}&(3,2) &
(3,2)\cr \sigma_{x}&(2,3) & (2,3)}.
\end{equation}
Though the outcome corresponding to any strategy profile of the
output game can be obtained through a suitable strategy profile in
the input game, the output game is not equivalent to the classical
one. The outcomes greater that 3 are not available in the output
game and the only reasonable result of the game is $(3,2)$---the
result that is never chosen by rational players in the input
game.

\paragraph{Insufficient range of players' actions} The two element set $\{\mathds{1},
\sigma_{x}\}$ of players' strategies seems to be appropriate in the
MW scheme. The identity $\mathds{1}$ and the Pauli matrix
$\sigma_{x}$ are supposed to represent classical moves,
while the initial quantum state $|\psi_{\mathrm{in}}\rangle$ plays
the role of a joint strategy. However, such a setting can lead to
a situation in which a player has no influence on the outcome of
the output game.
 This can be easily observed with example~(\ref{diagram1}). The second qubit is prepared in the superposition $(|0\rangle +
|1\rangle)/\sqrt{2}$. Then, no matter what operation (any
probability distribution on $\mathds{1}$ and $\sigma_{x}$) the
second player performs on his own qubit, he cannot change the
state and therefore he cannot influence the outcome of the game. In
this case the second player would need operators
\begin{equation}
\frac{1}{\sqrt{2}}\begin{pmatrix} 1 & 1 \\ 1 & -1
\end{pmatrix} \quad \mbox{and} \quad \frac{1}{\sqrt{2}}\begin{pmatrix} -1 & 1 \\ 1
&1
\end{pmatrix}\end{equation} in order to obtain the states
$|0\rangle$ and $|1\rangle$, respectively. Only then would
the outcomes of the input game in (\ref{diagram1}) be
achievable.

In fact, any separable state other than $|ij\rangle$ can lead
through the MW approach to output games which are not equivalent
to classical ones. The way to remove this property (keeping the
general MW approach unchanged) is to broaden the range of unitary
operations for players. Since a qubit in a separable state takes
the effective form
\begin{equation}\label{qubit}
\cos\frac{\theta}{2}|0\rangle +
e^{i\varphi}\sin\frac{\theta}{2}|1\rangle,
\end{equation}
two-parameter unitary operators are required to obtain the states
$|0\rangle$ and $|1\rangle$ from state (\ref{qubit}). For this
reason, let us consider unitary operators
\begin{equation} \label{twoparameter}
U(\theta, \varphi) = \displaystyle \begin{pmatrix}
\cos\frac{\theta}{2} &
e^{-i\varphi}\sin\frac{\theta}{2}\\e^{i\varphi}\sin\frac{\theta}{2}
& -\cos\frac{\theta}{2}
\end{pmatrix}
\end{equation}
and assume that each player $i$'s strategies are triples
\begin{equation}\mbox{for player 1:} ~~
\{(p,\theta_{1}, \varphi_{1}) \colon 0\leqslant p \leqslant 1, 0
\leqslant \theta_{1} \leqslant \pi, 0 \leqslant \varphi_{1}
\leqslant 2\pi \},
\end{equation}
\begin{equation}\mbox{for player 2:} ~~
\{(q,\theta_{2}, \varphi_{2}) \colon 0\leqslant q \leqslant 1, 0
\leqslant \theta_{2} \leqslant \pi, 0 \leqslant \varphi_{2}
\leqslant 2\pi \}.
\end{equation}
Next, we define the final state
\begin{align}
\rho_{\mathrm{fin}} &= pqU_{1}\otimes
U_{2}\rho_{\mathrm{in}}U_{1}^{\dag}\otimes U_{2}^{\dag} +
p(1-q)U_{1} \otimes V_{2}\rho_{\mathrm{in}}U_{1}^{\dag} \otimes V_{2}^{\dag}\nonumber\\
&\quad + (1-p)qV_{1} \otimes U_{2}\rho_{\mathrm{in}} V_{1}^{\dag}
\otimes U_{2}^{\dag}+ (1-p)(1-q)V_{1} \otimes
V_{2}\rho_{\mathrm{in}} V_{1}^{\dag} \otimes V_{2}^{\dag},
\end{align}
where
\begin{equation}\label{operators}
U_{i} := U(\theta_{i}, \varphi), \quad V_{i} = U(\pi-\theta_{i},
\varphi_{i} - \pi).
\end{equation}
Then for any separable pure state
\begin{equation}
|\psi_{\mathrm{in}}\rangle =
e^{i\delta}\left(\cos\frac{\theta_{1}}{2}|0\rangle +
e^{i\varphi_{1}}\sin{\frac{\theta_{1}}{2}}|1\rangle\right) \otimes
\left(\cos\frac{\theta_{2}}{2}|0\rangle +
e^{i\varphi_{2}}\sin{\frac{\theta_{2}}{2}}|1\rangle\right),
\end{equation}
each player $i$ can select appropriate $\theta_{i}$ and
$\varphi_{i}$ such that the operators (\ref{operators}) create the
final state
\begin{equation}
\rho_{\mathrm{fin}} = pq|00\rangle \langle00| + p(1-q)|01\rangle
\langle 01| + (1-p)q|10\rangle \langle 10| + (1-p)(1-q)|11\rangle
\langle 11|,
\end{equation}
which implies the input (classical) game in the MW scheme.

The extension (\ref{twoparameter})-(\ref{operators}) removes the
players' powerlessness against maximal superpositions.
Moreover, it guarantees that the classical game can be
reconstructed in the case of separable states. The solution,
however, seems a bit artificial. One may question why the players
are not allowed to use the full range of unitary operators or why
a probability distribution over the two-parameter operators is not
given by a probability density function. Moreover, the refinement
(\ref{twoparameter})-(\ref{operators}) does not solve the
following problem:

\paragraph{Choice of the initial state} The choice of the initial
state has a significant impact on the output game. However, the
MW protocol does not allow the players to take a part
in choosing the initial state. It causes many undesirable features
of the scheme. In particular, there is a possibility that the
output game induced by some fixed initial state could favor one
of the players and be unjust to the other one, even if the
strategic positions in the input game are completely identical. In
other words, it would be the case that a player is not
satisfied by the initial state and she would be willing to change
the initial state if she were allowed to.  Since the initial state
is supposed to be the players' joint strategy, the players should have
some influence in choosing the initial state, otherwise it stands
in contradiction with game-theoretical sense of strategy being an
element chosen by a player.

Furthermore, the fact that the initial state is treated as the
joint strategy ought to be taken into consideration in
game-theoretical solution concepts. For example, in a Nash
equilibrium no player gains by deviating from the equilibrium
strategy. Therefore, each player should have the possibility of
rejecting the initial state $|\psi_{\mathrm{in}}\rangle$ if the
`classical' state $|00\rangle$ would increase his payoff.

\section{A new model based on the MW scheme} Let us modify the MW
scheme to obtain a protocol which is free from the flaws listed in
the previous section. In order to do that let us redefine the
final state (\ref{staryfinalstate}) keeping the operator
(\ref{payoffoperator}) and payoff functional~(\ref{payofffunctional}) unchanged. We assume now that the players are additionally allowed to choose whether they play the
classical or quantum game. That is, each player decides to choose
either $\mathds{1}$ or $\sigma_{x}$ and at the same time whether
to perform the chosen local action on $|00\rangle$ or
$|\psi_{\mathrm{in}}\rangle$. Since the quantum state
$|\psi_{\mathrm{in}}\rangle$ is considered the joint strategy, we
assume that if the both players pick the action `quantum', the
chosen actions from the set $\{\mathds{1},\sigma_{x}\}$ are
performed on state $|\psi_{\mathrm{in}}\rangle$, otherwise the
players apply the local actions to state $|00\rangle$.

Denote by $C$ and $Q$ the actions `classical' and `quantum',
respectively. In such a scheme each player has four pure
strategies which we denote by $C\times \mathds{1}, C \times
\sigma_{x}, Q\times \mathds{1}$ and $Q \times \sigma_{x}$. Thus,
 the extended Marinatto-Weber (eMW) scheme is formally defined as
follows:
\begin{definition} \label{nowyskim}
If $|\psi_{\mathrm{in}}\rangle \in \mathbb{C}^2 \otimes
\mathbb{C}^2 $ represents a pure quantum state, $\tau_{1} =
(p_{1}, p_{2}, p_{3}, p_{4})$ and $\tau_{2} = (q_{1}, q_{2},
q_{3}, q_{4})$ are mixed strategies over $\{C\times \mathds{1}, C
\times \sigma_{x}, Q\times \mathds{1}, Q \times \sigma_{x}\}$ of
player 1 and 2, respectively, i.e. $p_{i}, q_{i} \geqslant 0$,
$\sum_{i}p_{i}=1$ and $\sum_{i}q_{i}=1$ and $X$ is as in
(\ref{payoffoperator}), then the extended final state
$\rho_{\mathrm{ext}}$ is given by
\begin{align} \label{newinitialstate}
\rho_{\mathrm{ext}} &= \mathds{1}\otimes \mathds{1}
\bigl((p_{1}q_{1} + p_{1}q_{3} + p_{3}q_{1})|00\rangle \langle 00|
+ p_{3}q_{3}|\psi_{\mathrm{in}}\rangle \langle
\psi_{\mathrm{in}}|\bigr)\mathds{1}\otimes \mathds{1} \nonumber\\
&\quad + \mathds{1}\otimes \sigma_{x}\bigl((p_{1}q_{2} +
p_{1}q_{4}+p_{3}q_{2})|00\rangle \langle 00| +
p_{3}q_{4}|\psi_{\mathrm{in}}\rangle \langle
\psi_{\mathrm{in}}|\bigr)\mathds{1}\otimes \sigma_{x} \nonumber\\
&\quad + \sigma_{x}\otimes \mathds{1} \bigl((p_{2}q_{1} +
p_{2}q_{3} + p_{4}q_{1})|00\rangle \langle 00| +
p_{4}q_{3}|\psi_{\mathrm{in}}\rangle \langle \psi_{\mathrm{in}}|
\bigr)\sigma_{x}\otimes \mathds{1} \nonumber\\ &\quad +
\sigma_{x}\otimes \sigma_{x} \bigl((p_{2}q_{2} + p_{2}q_{4} +
p_{4}q_{2})|00\rangle \langle 00| +
p_{4}q_{4}|\psi_{\mathrm{in}}\rangle \langle
\psi_{\mathrm{in}}|\bigr)\sigma_{x} \otimes \sigma_{x}
\end{align}
and the payoff pair is $\pi(\tau_{1}, \tau_{2}) = \mathrm{tr}(X
\rho_{\mathrm{ext}})$.
\end{definition}
Let us see how formula (\ref{newinitialstate}) works. If, for
example, player 1 chooses operation $\mathds{1}$ and the actions
`classical' and `quantum' with equal probability (which means
choosing $C \times \mathds{1}$ and $Q \times \mathds{1}$ with
equal probability) and player 2 picks $Q \times \sigma_{x}$, then
$p_{1}=p_{3} = 1/2$ and $q_{4} = 1$. It implies that formula
(\ref{newinitialstate}) takes on $\rho_{\mathrm{ext}} =
\bigl(|01\rangle \langle 01| + \mathds{1}\otimes \sigma_{x}
|\psi_{\mathrm{in}}\rangle \langle \psi_{\mathrm{in}}| \mathds{1}
\otimes \sigma_{x}\bigr)$/2.

Like in the MW scheme there exists a convenient way to express the
eMW scheme using bimatrix form, in this case by means of a
$4\times 4$ bimatrix game. The individual entries of the bimatrix
can be obtained by determining the payoffs $\pi(\tau_{1},
\tau_{2})$ for each of the 16 possible pure profiles
\begin{equation}
(\tau_{1}, \tau_{2}) \in \bigl\{\{C, Q\} \times \{\mathds{1},
\sigma_{x}\}\bigr\}_{1} \times \bigl\{\{C, Q\} \times
\{\mathds{1}, \sigma_{x}\}\bigr\}_{2}.
\end{equation}
Another way to obtain the bimatrix form is to use diagram
(\ref{diagram}). Indeed, formula (\ref{newinitialstate}) says that
players play the game on the right of diagram (\ref{diagram}) if
they both have chosen $Q$. In this case $p_{1} = p_{2} = q_{1} =
q_{2} = 0$ and then formula (\ref{newinitialstate}) coincides with
the final state (\ref{staryfinalstate}). In the other cases (i.e.
given that $p_{3}q_{3} = p_{3}q_{4} = p_{4}q_{3} = p_{4}q_{4} =
0$) the players play the classical game (on the left of diagram
(\ref{diagram})) because they perform $\mathds{1}$ and
$\sigma_{x}$ on state $|00\rangle \langle 00|$ then. We thus
obtain the following bimatrix counterpart:
\begin{equation}\label{4x4}
\bordermatrix{&C \times \mathds{1} & C \times \sigma_{x} & Q\times
\mathds{1} & Q \times \sigma_{x} \cr  C \times
\mathds{1}&(a_{00},b_{00}) & (a_{01}, b_{01}) & (a_{00},b_{00}) &
(a_{01},b_{01})\cr C\times \sigma_{x}&(a_{10},b_{10}) & (a_{11},
b_{11}) & (a_{10},b_{10}) & (a_{11},b_{11}) \cr
 Q\times \mathds{1}&(a_{00},b_{00}) & (a_{01}, b_{01}) &
(\alpha_{00},\beta_{00}) & (\alpha_{01},\beta_{01}) \cr Q \times
\sigma_{x}&(a_{10},b_{10}) & (a_{11}, b_{11}) &
(\alpha_{10},\beta_{10}) & (\alpha_{11},\beta_{11}) },
\end{equation}
where $(\alpha_{ij}, \beta_{ij})$ are as in (\ref{pary}). The eMW
scheme generalizes the classical way of playing games. If
$|\psi_{\mathrm{in}}\rangle = |00\rangle$, the strategies from the
set $\{C,Q\} \times \mathds{1}$ and $\{C,Q\} \times \sigma_{x}$
are {\it equivalent} to each other, i.e.  they induce the same
payoff pairs for any pure strategy played by the opponent.
Therefore, the {\it quotient} game obtained by removing equivalent
strategies (leaving the one representative) coincides with the
original one in respect of the game outcomes and players'
strategic positions.

In general, the eMW scheme provides output games quite different
from ones induced by the MW scheme. Since a unilateral deviation
from the classical action $C$ does not make the players play on
$|\psi_{\mathrm{in}}\rangle$, a Nash equilibrium in an input game
(\ref{2x2game}) remains the equilibrium in the quantum
counterpart~(\ref{4x4}). Moreover, in a case that the players do
not find the quantum strategy mutually advantageous, the classical
equilibria are the only ones in game (\ref{4x4}). This property is
well illustrated by the following example:
\begin{example}\label{examplemarinatto}
\textup{Let us consider a $2\times2$ game given by bimatrix
\begin{equation}
\bordermatrix{&l & r \cr t & (\alpha, \beta) & (\gamma, \gamma)
\cr b & (\gamma,\gamma) & (\beta, \alpha)},  ~~\text{where}~~
\alpha, \beta, \gamma \in \mathbb{R} ~~\text{and}~~ \alpha
> \beta
> \gamma.
\end{equation}
This is the general form of a game called Battle of the Sexes (the
game on the left of diagram (\ref{diagram1}) is the particular
case). The example was used in \cite{mw} to show that the players
playing the quantum strategy $|\psi_{\mathrm{in}}\rangle =
(|00\rangle + |11\rangle)/\sqrt{2}$ have access to Nash
equilibria with better results than in the case of playing only
classical strategies (see also \cite{benjamin} and
\cite{frackiewiczbattle}). The eMW scheme shows, however, that
state $(|00\rangle + |11\rangle)/\sqrt{2}$ would never be chosen
by the players. Indeed, if the players are allowed to decide
whether to choose the entangled state or not, they face the
following game:}
\begin{equation} \label{przyklad4x4}
\bordermatrix{&C \times \mathds{1} & C \times \sigma_{x} & Q
\times \mathds{1} & Q \times \sigma_{x} \cr  C \times
\mathds{1}&(\alpha,\beta) & (\gamma, \gamma) & (\alpha,\beta) &
(\gamma,\gamma)\cr C\times \sigma_{x}&(\gamma,\gamma) &
(\beta,\alpha) & (\gamma,\gamma) & (\beta,\alpha) \cr
 Q\times \mathds{1}&(\alpha,\beta) & (\gamma, \gamma) &
\left(\displaystyle\frac{\alpha+\beta}{2},
\displaystyle\frac{\alpha+\beta}{2}\right) & (\gamma,\gamma) \cr Q
\times \sigma_{x}&(\gamma,\gamma) & (\beta, \alpha) &
(\gamma,\gamma) & \left(\displaystyle\frac{\alpha+\beta}{2},
\displaystyle\frac{\alpha+\beta}{2}\right) }
\end{equation}
\textup{Now, if the players were restricted to using only $Q\times
\{\mathds{1}, \sigma_{x}\}$, game (\ref{przyklad4x4}) would come
down to one associated with the MW scheme and strategy profiles:
$(Q \times \mathds{1}, Q \times \mathds{1}), (Q \times \sigma_{x},
Q \times \sigma_{x})$ and one in which each player plays $Q \times
\mathds{1}$ and $Q \times \sigma_{x}$ with equal probability would
constitute Nash equilibria. However, these profiles are no longer
equilibria in game (\ref{przyklad4x4}). The first player's best
response to $Q \times \mathds{1}$ of player 2 is $C \times
\mathds{1}$ instead of $Q \times \mathds{1}$. The second pure
profile is not an equilibrium for similar reasons. Next, for
probability distribution $\{1/2,1/2\}$ over $Q \times
\{\mathds{1}, \sigma_{x}\}$ played by one of the players, the best
response of the other player is $C\times \mathds{1}$. As a result,
since a Nash equilibrium is considered a necessary condition for a
strategy profile to be a reasonable one, the quantum strategy
$(|00\rangle + |11\rangle)/\sqrt{2}$ would not be chosen by the
players.} \textup{It does not have to indicate, however, that the
players cannot benefit from quantum strategies. Let us consider
similar (in terms of the MW scheme) state $(|01\rangle +
|10\rangle)/\sqrt{2}$. The MW approach through this state also
implies two pure Nash equilibria with outcome $(\alpha + \beta)/2$
for each player and the Nash equilibrium in mixed strategies.
However, only $(Q\times \mathds{1}, Q\times \sigma_{x})$ remains
an equilibrium if the players are able to decide whether to use
the quantum state or not since they play the game given by
bimatrix
\begin{equation} \label{przyklad4x4a}
\bordermatrix{&C \times \mathds{1} & C \times \sigma_{x} & Q
\times \mathds{1} & Q \times \sigma_{x} \cr  C \times
\mathds{1}&(\alpha,\beta) & (\gamma, \gamma) & (\alpha,\beta) &
(\gamma,\gamma)\cr C\times \sigma_{x}&(\gamma,\gamma) &
(\beta,\alpha) & (\gamma,\gamma) & (\beta,\alpha) \cr
 Q\times \mathds{1}&(\alpha,\beta) & (\gamma, \gamma) &
(\gamma, \gamma) & \left(\displaystyle\frac{\alpha+\beta}{2},
\displaystyle\frac{\alpha+\beta}{2}\right) \cr Q \times
\sigma_{x}&(\gamma,\gamma) & (\beta, \alpha) &
\left(\displaystyle\frac{\alpha+\beta}{2},
\displaystyle\frac{\alpha+\beta}{2}\right) & (\gamma, \gamma) }.
\end{equation}
The opposite players' preferences in the classical battle of the
sexes game (player 1 prefers $(C\times \mathds{1},C\times
\mathds{1})$, whereas player 2 prefers $(C\times
\sigma_{x},C\times \sigma_{x})$ if game is played classically)
make $(Q\times \mathds{1}, Q\times \sigma_{x})$ the most
reasonable profile.}
\end{example}
\section{The disadvantages of the MW approach versus the eMW
scheme} Since a quantum state $|\psi_{\mathrm{in}}\rangle$ is
considered to be a players' joint strategy, the two additional
actions: $C$ (classical) and $Q$ (quantum) appear to be a natural
extension of the MW scheme. It turns out that this extension is
sufficient to remove the disadvantages we listed in
section~\ref{disadvantages}. The problem of the powerlessness of
players' actions (when the qubits are the maximal superpositions)
is not a concern in the extended scheme. The action $C$ played by
just one of the players turns the game into a classical one. For
the same reason, each outcome of the classical game is still
available in the quantum counterpart (\ref{4x4}). The extended
scheme also removes, in some sense, the problem concerning the
choice of the initial state since the players can decide whether
to play the quantum state or not. The fact that the initial state
has to be chosen by an arbiter beforehand just means that there
are infinitely many possible quantum extensions of some classical
game. Every such an extension is associated with some initial
quantum state then.

\paragraph{Division of games into classical and quantum ones}
It has been proved that classical correlations can always be
associated with separable states \cite{gisin2} and any pure
entangled state violates some Bell inequality \cite{gisin,
popescu}.  So, an interesting problem is to examine if the eMW
scheme could give us a similar hierarchical structure in quantum
game theory.

First, let us assign to mixed state (\ref{newinitialstate}) a set
of pure quantum states
 by means of the following simple fact:
\begin{fact}
Let $X$ be as in (\ref{payoffoperator}). For any density operator
$\rho$ on $\mathbb{C}^2 \otimes \mathbb{C}^2$ there exists a pure
state $|\psi\rangle \in \mathbb{C}^2 \otimes \mathbb{C}^2$ and a
mixed state $\sum_{ij=0,1}\lambda_{ij}|ij\rangle \langle ij|$ such
that
\begin{equation}\label{equality}
\mathrm{tr}(X \rho) = \mathrm{tr}(X |\psi\rangle \langle \psi |) =
\mathrm{tr}\left(X \sum_{ij=0,1}\lambda_{ij}|ij\rangle \langle
ij|\right).
\end{equation}
\end{fact}
\begin{proof}
Let $\rho = \sum_{i,j,k,l=0,1}O_{ijkl}|ij\rangle \langle kl|$ be a
density operator. Putting state $|\psi\rangle = \sum_{i,j =
0,1}\sqrt{O_{ijij}}|ij\rangle$ and numbers $\lambda_{ij} =
O_{ijij}$ we obtain equality (\ref{equality}).
\end{proof}
Thus, we can always relate the final state $\rho_{\mathrm{ext}}$
to an equivalent pure state (with respect to the measurement $X$).
Such an assignment is necessary because specification of the MW
scheme allows players to obtain strictly non-classical results
with separable mixed states. For instance, separable state
$\rho_{\mathrm{ext}} = (|00\rangle \langle 00| + |11\rangle
\langle 11|)/2$ given by taking $p_{3} = 1$, $q_{1} = q_{3} = 1/2$
and $|\psi_{\mathrm{in}}\rangle = |11\rangle$ in
(\ref{newinitialstate}) induces
$\mathrm{tr}(X\rho_{\mathrm{ext}})=\bigl((a_{00},b_{00}) +
(a_{11},b_{11})\bigr)/2$---the outcome that is not available
with the use of classical mixed strategies. However, if we
restrict ourselves to pure states, only entangled pure state
$(|00\rangle + e^{i\varphi}|11\rangle)/\sqrt{2}$ (up to the phase
factor) generates that outcome. Note that the state
$|\psi_{\mathrm{in}}\rangle$ is just a part of
$\rho_{\mathrm{ext}}$. Thus, it does not have to be entangled so
that the final state generates a non-classical game.
\begin{example}\label{examplegraniepewnosc}
\textup{Let us consider the following bimatrix game:
\begin{equation}\label{graniepewnosc}
\bordermatrix{&l & r \cr t & (5, 5) & (0, 4) \cr b & (4,0) &
(2,2)}.
\end{equation}
The game exhibits a conflict between two reasonable types of
equilibria. The profile $(b,r)$ is a risk dominant equilibrium
\cite{harsai} (see also \cite{frackiewiczbattle}). If one of the
players is uncertain about the move of the other one, a risk
dominant strategy turns out to be a reasonable choice. On the
other hand, the profile $(t,l)$ is a payoff dominant equilibrium
for which both players receive much more.}

\textup{The uncertainty about the result of game
(\ref{graniepewnosc}) disappears in the quantum domain. The final
state $\rho_{\mathrm{ext}}$ for $|\psi_{\mathrm{in}}\rangle =
|11\rangle$ implies bimatrix
\begin{equation}\label{gra24}
\bordermatrix{&C \times \mathds{1} & C \times \sigma_{x} & Q\times
\mathds{1} & Q \times \sigma_{x} \cr  C \times \mathds{1}&(5,5) &
(0, 4) & (5,5) & (0,4)\cr C\times \sigma_{x}&(4,0) & (2,2) & (4,0)
& (2,2) \cr
 Q\times \mathds{1}&(5,5) & (0, 4) &
(2, 2) & (4,0) \cr Q \times \sigma_{x}&(4,0) & (2, 2) & (0,4) &
(5, 5)}
\end{equation}
The profile $(Q \times \mathds{1}, Q \times \mathds{1})$
corresponding to the risk-dominant outcome is no longer
equilibrium in (\ref{gra24}). If the players decide to play
quantum strategies, the equilibrium $(Q\times \sigma_{x}, Q\times
\sigma_{x})$ that provides the players with the payoff dominant
outcome is the unique reasonable profile.}
\end{example}
We hypothesize that the association with a pure separable state is
a sufficient condition on $\rho_{\mathrm{ext}}$ for a game to be
classical. It is not work in the case of the original MW scheme.
In fact, if $|\psi_{\mathrm{in}} \rangle$ is separable, we can
always find a separable pure state that is equivalent to
(\ref{staryfinalstate}) (with regard to
measurement~(\ref{payoffoperator})) but, as
diagram~(\ref{diagram1}) shows, the output game does not coincide
with the classical one. In the eMW scheme the players have a
greater impact on the final state $\rho_{\mathrm{ext}}$ despite
the same local operators $\mathds{1}$ and $\sigma_{x}$. It follows
from the fact that the players do not perform $\mathds{1} \otimes
\mathds{1}$, $\mathds{1} \otimes \sigma_{x}$, $\sigma_{x} \otimes
\mathds{1}$ and $\sigma_{x} \otimes \sigma_{x}$ on the whole
quantum state but on the appropriate parts of it. This distinction
is relevant as it allows us to formulate the following
proposition:
\begin{proposition} \label{propozycja}
A quotient game of a game specified by definition~\ref{nowyskim}
coincides with the input game (\ref{2x2game}) if and only if for
any players' strategies ${\tau_{1}}$ and $\tau_{2}$, the final
state (\ref{newinitialstate}) satisfies equation (\ref{equality})
for some separable pure state $|\psi\rangle \in
\mathbb{C}^2\otimes \mathbb{C}^2$.
\end{proposition}
\begin{proof}
Let us first assume that the quotient game is equal to the input
game (\ref{2x2game}). In this case the pairs of strategies $C
\times \mathds{1}, Q \times \mathds{1}$ and $C \times \sigma_{x},
Q \times \sigma_{x}$ have to be equivalent to each other which
imposes $|\psi_{\mathrm{in}}\rangle = |00\rangle$ (up to the phase
factor) on the final state (\ref{newinitialstate}) and therefore
\begin{equation}
\rho_{\mathrm{ext}} = \bigl((p_{1} + p_{3})|0\rangle \langle 0| +
(p_{2} + p_{4})|1\rangle \langle 1|\bigr) \otimes \bigl((q_{1} +
q_{3})|0\rangle \langle 0| + (q_{2} + q_{4})|1\rangle \langle 1|
\bigr).
\end{equation}
Then for any $p_{i}, q_{i}$, a separable state
\begin{equation}
|\psi\rangle = \bigl(\sqrt{p_{1} + p_{3}}|0\rangle + \sqrt{p_{2} +
p_{4}}|1\rangle\bigr) \otimes \bigl(\sqrt{q_{1} + q_{3}}|0\rangle
+ \sqrt{q_{2} + q_{4}}|1\rangle\bigr)
\end{equation}
satisfies the equation $\mathrm{tr}(X \rho_{\mathrm{fin}}) =
\mathrm{tr}(X|\psi\rangle \langle \psi|)$.

On the other hand, let us determine the general form of the final
state $\rho_{\mathrm{ext}}$. The equation (\ref{equality}) shows
that
 we can replace
$|\psi_{\mathrm{in}}\rangle \langle \psi_{\mathrm{in}}|$ with some
mixed state $\sum_{i,j = 0,1} \eta_{ij}|ij\rangle \langle ij|$
without loss of generality. Thus the general final state can be
written as follows:
\begin{align}\label{stateproof}
\rho_{\mathrm{ext}} &= (p_{1}q_{1} + p_{1}q_{3} + p_{3}q_{1} +
p_{3}q_{3}\eta_{00} + p_{3}q_{4}\eta_{01} + p_{4}q_{3}\eta_{10} +
p_{4}q_{4}\eta_{11})|00\rangle \langle 00|\nonumber \\
&\quad+(p_{1}q_{2} + p_{1}q_{4} + p_{3}q_{2} +
p_{3}q_{3}\eta_{01}+ p_{3}q_{4}\eta_{00} + p_{4}q_{3}\eta_{11} +
p_{4}q_{4}\eta_{10})|01\rangle \langle 01| \nonumber \\
&\quad+(p_{2}q_{1} + p_{2}q_{3} + p_{4}q_{1} + p_{3}q_{3}\eta_{10}
+ p_{3}q_{4}\eta_{11} + p_{4}q_{3}\eta_{00} +
p_{4}q_{4}\eta_{01})|10\rangle \langle 10| \nonumber \\
&\quad+ (p_{2}q_{2} + p_{2}q_{4} + p_{4}q_{2} +
p_{3}q_{3}\eta_{11} + p_{3}q_{4}\eta_{10} + p_{4}q_{3}\eta_{01} +
p_{4}q_{4}\eta_{00})|11\rangle \langle 11|.
\end{align}
Now, note that if $|\psi\rangle$ is separable, the third term in
the equation (\ref{equality}) satisfies a separability condition
$\lambda_{00}\lambda_{11} = \lambda_{01} \lambda_{10}$. Since for
any $p_{i}, q_{i}$, state~(\ref{stateproof}) is associated with
some separable state, the separability condition gives equations
\begin{equation} \label{3equations}
\left\{\begin{array}{ll} (2+\eta_{10})\eta_{01} =
(1+\eta_{00})\eta_{11} &\text{for}~~ p_{1} = p_{4} = q_{1} =
q_{3} = \frac{1}{2};\\
(2+ \eta_{01})\eta_{10} = (1+\eta_{00})\eta_{11} &\text{for}~~
p_{1} = p_{3} = q_{1} = q_{4} = \frac{1}{2};\\
(3+\eta_{00})\eta_{11} = \eta_{01}\eta_{10} &\text{for}~~p_{1} =
p_{3} = q_{1} = q_{3} = \frac{1}{2}.
\end{array}\right.
\end{equation}
Together with the normalization condition $\sum_{i,j}\eta_{i,j} =
1$, equations (\ref{3equations}) imply $\eta_{00}=1$. As a result,
state (\ref{stateproof}) determines a game which corresponds to
the input one (up to a quotient game).
\end{proof}
Proposition~\ref{propozycja} tells us that
$|\psi_{\mathrm{in}}\rangle \ne |00\rangle$ in
(\ref{newinitialstate}) provides the players with nonclassical
strategies and, as examples~\ref{examplemarinatto} and
\ref{examplegraniepewnosc} show, may significantly influence the
game result. The non-classical result can be obtained only if the
final state relates to some entangled pure state. Thus, like in
the case of theory of quantum correlation, there exists a
hierarchical structure for quantum games in regard to separable
and entangled states.
\section{Further generalization of the scheme}
\paragraph{More quantum strategies}
The eMW approach can be generalized to allow for more than one
joint strategy $|\psi_{\mathrm{in}}\rangle$. Denote by
$|\psi_{1}\rangle, \dots, |\psi_{n}\rangle \in \mathbb{C}^2
\otimes \mathbb{C}^2$ the $n$ available joint quantum strategies
and by $Q_{1}, \dots, Q_{n}$ the corresponding `quantum' actions.
Then, together with the local actions $\mathds{1}$ and
$\sigma_{x}$, each player has the $2(n+1)$-element set of
strategies $\{C, Q_{1}, \dots, Q_{n}\} \times \{\mathds{1},
\sigma_{x}\}$. Similarly to (\ref{4x4}), we assume that the
quantum strategy $|\psi_{i}\rangle$ is played only if both players
choose $Q_{i}$. Otherwise, the players perform the local
operations on $|00\rangle \langle 00|$. As a result, the output
game can be characterized by matrix
\begin{equation}\label{ostatnia4x4}
\bordermatrix{&C \times \{\mathds{1},\sigma_{x}\} & Q_{1} \times
\{\mathds{1},\sigma_{x}\} & \cdots & Q_{n} \times \{\mathds{1},
\sigma_{x}\} \cr \;\, C \times \{\mathds{1}, \sigma_{x}\}
&|00\rangle \langle 00| & |00\rangle \langle 00| & \cdots &
|00\rangle \langle 00| \cr Q_{1}\times \{\mathds{1}, \sigma_{x}\}
& |00\rangle \langle 00| & |\psi_{1}\rangle \langle \psi_{1}| &
\cdots& |00\rangle \langle 00| \cr ~~~~~\vdots &\vdots & \vdots &
\ddots & \vdots \cr Q_{n} \times \{\mathds{1}, \sigma_{x}\} &
|00\rangle \langle 00|&|00\rangle \langle 00|& \cdots &
|\psi_{n}\rangle \langle \psi_{n}|}.
\end{equation}
Formally, the construction of the appropriate final state is as
follows. Denote by $(p_{1}, p_{2}, \dots, p_{2(n + 1)})$ and
$(q_{1}, q_{2}, \dots, q_{2(n + 1)})$ the mixed strategies of
player 1 and 2, respectively. Probabilities $p_{1}$ and $p_{2}$
($q_{1}$ and $q_{2}$) are associated with actions $C \times
\mathds{1}$ and $C \times \sigma_{x}$, the pair $p_{2i + 1}, p_{2i
+ 2}$ (the pair $q_{2i + 1}, q_{2i + 2}$) concerns $Q_{i} \times
\mathds{1}$ and $Q_{i} \times \sigma_{x}$ for $i = 1,2,\dots,n$.
For each $i,j \in \{0,1, \dots, n\}$, we define a density operator
$\rho_{ij}$,
\begin{align}
\rho_{ij} &= p_{2i + 1}q_{2j +
1}\mathds{1}\otimes\mathds{1}|\phi_{ij}\rangle \langle \phi_{ij}|
\mathds{1}\otimes \mathds{1} +  p_{2i + 1}q_{2j +
2}\mathds{1}\otimes\sigma_{x}|\phi_{ij}\rangle \langle \phi_{ij}|
\mathds{1}\otimes \sigma_{x} \nonumber \\
&\quad+p_{2i + 2}q_{2j +
1}\sigma_{x}\otimes\mathds{1}|\phi_{ij}\rangle \langle \phi_{ij}|
\sigma_{x}\otimes \mathds{1} + p_{2i + 2}q_{2j +
2}\sigma_{x}\otimes\sigma_{x}|\phi_{ij}\rangle \langle \phi_{ij}|
\sigma_{x}\otimes \sigma_{x},
\end{align}
where
\begin{equation}
|\phi_{ij}\rangle =\left\{\begin{array}{c} |\psi_{i}\rangle\quad
\mbox{if}\quad i = j \ne 0;\\ |00\rangle\quad \mbox{if}\quad
\mbox{otherwise}.
\end{array}\right.
\end{equation}
Then the general final state corresponding to (\ref{ostatnia4x4})
is $\rho_{\mathrm{ext}} = \sum^{n}_{i,j = 0}\rho_{ij}$.
\paragraph{More local strategies} We showed in
\cite{commentfracor} how to construct the scheme for $n\times m$
bimatrix games according to the MW approach. In this case, player
1 (player 2) has $n$ operators $U_{i}$ ($m$ operators $V_{j}$)
defined on $\mathbb{C}^n$ ($\mathbb{C}^m$) that act on basis
states $\{|0\rangle, |1\rangle, \dots, |n-1\rangle\}$
($\{|0\rangle, |1\rangle, \dots, |m-1\rangle\}$) as follows:
\begin{equation}
\begin{array}{ll}
U_{0}|i\rangle = |i\rangle & V_{0}|i\rangle = |i\rangle, \\ U_{1}|i\rangle = |i + 1 \bmod n\rangle & V_{1}|i\rangle = |i + 1 \bmod m\rangle, \\
~~~~~~~~~~~~~~~~\vdots  \\ U_{n-1}|i\rangle = |i + (n-1) \bmod
n\rangle & V_{m-1}|i\rangle = |i + (l-1) \bmod m \rangle.
\end{array}
\end{equation}
The construction of the final state $\rho_{\mathrm{ext}}$ is
analogous to definition~\ref{nowyskim} but the set of pure
strategies for player 1 and 2 is now $\{C, Q\} \times \{U_{0},
\dots, U_{n-1}\}$ and $\{C, Q\} \times \{V_{0}, \dots, V_{m-1}\}$,
respectively.
\section{Conclusions}
During the last fourteen years of research into quantum games, the
MW \cite{mw} and EWL~\cite{ewl}  concepts have become the most
commonly used quantum schemes for $2 \times 2$ games. However, the
undesirable features of the MW scheme question if it may reflect a
real quantum game. Our protocol is closer to the EWL concept. In
both the EWL and eMW scheme there is the need for entangled states
to generate nonclassical results. Next, turning the quantum game
into classical one can be obtained just by the restriction of
players strategies. It is also possible to study the case where
only one of the players has access to the quantum strategy (in the
EWL it is performed by the appropriate restriction of the player's
strategy set). It could be done by allowing only one of players to
decide whether to play the state $|00\rangle$ of
$|\psi_{\mathrm{in}}\rangle$, i.e., only one of the players has
access to the actions $C$ and $Q$.

At the same time, the eMW scheme is much simpler that the EWL
scheme for studying high-dimensional bimatrix games. Thus, it may
constitute an alternative to the EWL scheme.
\section*{Acknowledgments}
The project was supported by the Polish National Science
Center under the project DEC-2011/03/N/ST1/02940.

\end{document}